		\newcommand\FINAL[1]{} 

\providecommand\apxparameter{appendix=inline}

\documentclass[11pt,a4paper]{article}
\pdfoutput=1

\usepackage[T1]{fontenc}
\usepackage[utf8]{inputenc}
\usepackage[UKenglish]{babel}
\usepackage{lmodern}
\usepackage{microtype}
\usepackage[margin=2.5cm]{geometry}
\usepackage{amsmath,amssymb,amsthm}
\usepackage{graphicx}
\usepackage{xcolor}
\usepackage{url}
\usepackage{tikz}
\usepackage{versions}

\usepackage[\apxparameter]{apxproof}
\newtheoremrep{theorem}{Theorem}
\newtheoremrep{claim}[theorem]{Claim}
\newtheoremrep{proposition}[theorem]{Proposition}
\newtheoremrep{lemma}[theorem]{Lemma}
\newtheoremrep{corollary}[theorem]{Corollary}
\newtheoremrep{openproblem}[theorem]{Open Problem}
\newtheorem{observation}[theorem]{Observation}
\theoremstyle{definition}
\newtheorem{definition}[theorem]{Definition}
\theoremstyle{plain}
\newtheorem{property}[theorem]{Property}
\newtheorem{problem}{Problem}
\newenvironment{claimproof}{\begin{proof}[Proof of the claim]}{\end{proof}}

\graphicspath{{./figures/}}
\title{Emergency Vertex Cover}
\author{Eric Angel\thanks{Universit\'e Paris-Saclay, Univ Evry, IBISC, France. \texttt{eric.angel@univ-evry.fr}}
\and Evangelos Bampas\thanks{Universit\'e Paris-Saclay, CNRS, LISN, 91400 Orsay, France. \texttt{bampas@lisn.fr}}
\and Evripidis Bampis\thanks{LIP6, Sorbonne Universit\'e, Paris, France. \texttt{evripidis.bampis@lip6.fr}}
\and Vincent Chau\thanks{Universit\'e Paris-Saclay, Univ Evry, IBISC, France. \texttt{vincent.chau@univ-evry.fr}}
\and Johanne Cohen\thanks{Universit\'e Paris-Saclay, CNRS, LISN, 91400 Orsay, France. \texttt{johanne.cohen@lisn.upsaclay.fr}}
\and Alexander Kononov\thanks{Sobolev Institute of Mathematics SB RAS, Russia. \texttt{alvenko@math.nsc.ru}}
\and Yizheng Zhang\thanks{Research Centre for Operations Research and Statistics, KU Leuven, Belgium. \texttt{yizheng.zhang@kuleuven.be}}}
\date{}

\usepackage{algorithmicx}
\usepackage{algpseudocode}

\DeclareMathOperator*{\argmin}{arg\,min}
\DeclareMathOperator*{\argmax}{arg\,max}

\usepackage{framed}
\usepackage{marginnote}

\newcommand{\OPT}{\ensuremath\mathrm{OPT}}

\newcommand{\probname}{\ensuremath\mathsf{Emergency\text{\ }Vertex\text{\ }Cover}}
\newcommand{\probnameshort}{\ensuremath\mathsf{Em\text{-}VC}}
\newcommand{\probnamecont}{\ensuremath\mathsf{Continuous\ Emergency\text{\ }Vertex\text{\ }Cover}}
\newcommand{\probnamecontshort}{\ensuremath\mathsf{CEm\text{-}VC}}
\newcommand{\NP}{\ensuremath\mathsf{NP}}
\newcommand{\SAT}{\ensuremath\mathsf{SAT}}

\newcommand{\broadcast}{\ensuremath\mathsf{Broadcast\text{\ }Domination}}

\newcommand{\msrdc}{\ensuremath\mathsf{MSRDC}}

\newcommand{\mincover}{\ensuremath\mathsf{Minimum\text{\ }Cost\text{\ }Covering}}
\newcommand{\mincovershort}{\ensuremath\mathsf{MCC}}
\newcommand{\diam}{\ensuremath\mathrm{diam}}
\newcommand{\power}{\ensuremath \mathbf{p}}
\newcommand{\edge}[2]{\{#1,#2\}}
\newcommand{\cost}[1]{\ensuremath\mathrm{cost}_{#1}}

\newcommand{\calE}{\ensuremath\mathcal{E}}
\newcommand{\calF}{\ensuremath\mathcal{F}}
\newcommand{\calI}{\ensuremath\mathcal{I}}
\newcommand{\calO}{\ensuremath\mathcal{O}}
\newcommand{\calP}{\ensuremath\mathcal{P}}
\newcommand{\calS}{\ensuremath\mathcal{S}}

\newcommand{\bbN}{\mathbb{N}}
\newcommand{\bbQ}{\mathbb{Q}}
\newcommand{\bbR}{\mathbb{R}}
\renewcommand{\star}{*}

\usepackage{hyperref}

\begin{document}

\maketitle

\begin{abstract}
The ${\ensuremath\mathsf{Minimum\text{\ }Vertex\text{\ }Cover}}$ problem is a fundamental combinatorial optimization problem, aiming to identify a minimum subset of vertices in a graph such that every edge is incident to at least one vertex in this subset. Among its variants, the ${\ensuremath\mathsf{Min\text{-}Power\text{-}Cover}}$ problem stands out due to its practical applications, such as camera placement at intersections: in an edge-weighted graph, an edge is covered if one of its endpoints is assigned a power value at least as large as the edge’s weight.

In this paper, we introduce the ${\ensuremath\mathsf{Emergency\text{\ }Vertex\text{\ }Cover}}$ (${\ensuremath\mathsf{Em\text{-}VC}}$) problem where
an edge may be covered not only by its endpoints, but also by a distant vertex, provided the vertex is given sufficient power to ``cover'' the cumulative weight of the edges along a shortest path to one of the edge’s endpoints plus the weight of the edge. 
${\ensuremath\mathsf{Em\text{-}VC}}$ is motivated by different practical scenarios, e.g. the need for urban disaster response, where 
ensuring accessibility to all road segments (edges of the graph) is crucial for effective aid delivery.

We prove that ${\ensuremath\mathsf{Em\text{-}VC}}$ is NP-hard, derive lower bounds, and design a polynomial-time algorithm for its continuous version. Moreover, we present a $\tfrac{4}{3}$-approximation algorithm for the discrete case and identify several special graph classes for which the problem can be solved in polynomial time.

\end{abstract}

\newpage
\section{Introduction}

The ${\ensuremath\mathsf{Minimum\text{\ }Vertex\text{\ }Cover}}$ problem is one of the most fundamental combinatorial optimization problems. Given a graph $G=(V,E)$, 
the goal is to determine a subset of vertices $C \subseteq V$ of minimum cardinality such that for every edge $e \in E$, 
there is at least one endpoint in $C$. Several variants of the ${\ensuremath\mathsf{Minimum\text{\ }Vertex\text{\ }Cover}}$ problem have been studied in the literature. 

One variant of particular interest to our work is the ${\ensuremath\mathsf{Min\text{-}Power\text{-}Vertex\text{-}Cover}}$ problem~\cite{angel2015min}, motivated by applications such as placing cameras at road intersections in a city. In ${\ensuremath\mathsf{Min\text{-}Power\text{-}Vertex\text{-}Cover}}$, we are given an edge-weighted graph, and  an edge $e$ is considered covered  if at least one of its endpoints has a valuation (power) greater than or equal to the weight of $e$. The objective is to determine both a vertex cover and a power assignment to its vertices so as to minimize the total power. More generally, in~\cite{angel2015min}, ${\ensuremath\mathsf{Min\text{-}Power\text{-}Cover}}$ variants of classical graph problems, including Asymmetric Vertex Cover, Minimum Cut, Spanning Tree, and $(s,t)$-Path, have been considered. A natural question is how the complexity of ${\ensuremath\mathsf{Min\text{-}Power\text{-}Cover}}$ variants  compares to that of their corresponding classical graph problems. Interestingly, some problems remain as difficult as their classical versions, while others become much more difficult in the ${\ensuremath\mathsf{Min\text{-}Power\text{-}Cover}}$ setting.

In this paper, we introduce the $\probname$ ($\probnameshort$) problem. In the $\probnameshort$ problem, we are given an edge-weighted graph and an edge $\edge{u}{w}$ may be \emph{covered} by a vertex $v$ (not necessarily an endpoint of the edge) if the power assigned to $v$ is at least the total weight of the edges along a shortest path (measured with respect to edge weights) either from $v$ to $u$ passing through $w$, or from $v$ to $w$ passing through $u$. Unlike ${\ensuremath\mathsf{Min\text{-}Power\text{-}Vertex\text{-}Cover}}$, which is a direct generalization of ${\ensuremath\mathsf{Minimum\text{\ }Vertex\text{\ }Cover}}$ and hence NP-hard, $\probnameshort$ does not obviously generalize a known NP-hard problem, and its complexity is an interesting question. However, the motivation for studying this problem comes not only from this theoretical question but also from applications such as disaster response in urban networks, supply chain management, and utility networks. For instance, after a natural disaster, emergency response teams must ensure that all roads (edges in the graph) remain accessible to deliver aid and evacuate people. 
Therefore, the objective is to strategically assign resources (e.g., fuel or supplies) to emergency stations so that every road (edge) can be accessed by at least one emergency response team.

Another related problem is $\broadcast$ in graphs, first introduced and studied in 2001 by Erwin~\cite{erwin2001cost,erwin2004dominating}.
It captures the problem of positioning radio transmitters with different effective radiated powers. In this model, the vertices represent broadcast sites and each vertex $v$ is assigned a broadcast strength $f(v)\in \{0,1,\ldots,\diam(G)\}$, where $\diam(G)$ is the diameter of~$G$. A vertex $u$ hears the broadcast from $v$ if its shortest-path distance $d(u,v)$ from~$v$ satisfies $d(u,v)\leq f(v)$. A function $f$ is a dominating broadcast if every vertex hears at least one broadcast. The cost of a broadcast is the sum of all strengths, and the broadcast domination number $\gamma_b(G)$ is the minimum cost of a dominating broadcast (see the survey~\cite{henning2021broadcast} for an overview).
The $\probnameshort$ problem can be viewed as a strengthened form of broadcast domination. Indeed, requiring that both endpoints of every edge be covered by the same broadcast vertex forces each edge to be served by a single source, yielding precisely the $\probnameshort$ model. Hence, every feasible solution to $\probnameshort$ is a valid dominating broadcast, whereas the converse is not true, since in standard broadcast domination the two endpoints of an edge may be covered by different vertices. This additional structural constraint has a significant impact on complexity. In general graphs, an optimal dominating broadcast can be computed in $O(n^6)$ time~\cite{heggernes2006optimal}, recently improved to $O(n^5)$~\cite{papadopoulos2026broadcast}, while in the $\probnameshort$ problem that we study, minimizing the total power is NP-hard.

\textbf{Related work.} \quad 
The ${\ensuremath\mathsf{Min\text{-}Power\text{-}Vertex\text{-}Cover}}$ problem was introduced in \cite{angel2015min}, where a $2$-approximation algorithm is presented. Additionally, a $3$-approximation algorithm for the ${\ensuremath\mathsf{Min\text{-}Power\text{-}Cover}}$ variant of the metric Traveling Salesman Problem (TSP) and a $2$-approximation algorithm for the ${\ensuremath\mathsf{Min\text{-}Power\text{-}Cover}}$ variant of the cycle cover problem have been presented in~\cite{angel2015min}. The problem has further been studied from the perspective of parameterized complexity in~\cite{angel2018parameterized}.
The authors proposed fixed-parameter tractable (FPT) algorithms when the problem is parameterized by the value of the {\em optimal total power} $P$, or by the {\em number of vertices} $k$ {\em receiving positive power}. They also showed that the problem becomes significantly more difficult than the classical ${\ensuremath\mathsf{Minimum\text{\ }Vertex\text{\ }Cover}}$ problem when parameterized by the {\em graph's treewidth} $t$. Specifically, they proved that, unless the Exponential Time Hypothesis (ETH) fails, there is no algorithm running in time $n^{o(t)}$. To overcome this hardness, they developed an FPT approximation scheme that computes a $(1 + \epsilon)$-approximation of the optimal solution in time that is FPT in parameters $t$ and $1/\epsilon$.

Note that the ${\ensuremath\mathsf{Min\text{-}Power\text{-}Cover}}$ model differs from the ${\ensuremath\mathsf{Min\text{-}Power}}$ model discussed in various works \cite{DBLP:journals/winet/AlthausCMPTZ06,DBLP:conf/esa/Grandoni12,DBLP:journals/tcs/KirousisKKP00}. In the ${\ensuremath\mathsf{Min\text{-}Power}}$ model, a specific requirement is that to cover a given edge, the power allocated to each of its endpoints must be at least the cost of that edge. This means that both endpoints must reach the power threshold of the edge to support its connection, which is not required in our model. \textcolor{black}{The ${\ensuremath\mathsf{Min\text{-}Power\text{-}Cover}}$ model may allow for more flexibility in power allocation while still ensuring connectivity.}
Additional related results can be found in~\cite{improved,cohen,fujito}. Prize-collecting variants have been studied in a geometric setting~\cite{dai2022approximation,liu2022primal}: sensors in the plane are assigned radii, a radius $r$ costs $r^\alpha$, and uncovered points incur a penalty. There the covered objects are points rather than edges, and the metric is Euclidean, so these models only loosely relate to ours.

In~\cite{drexler2023} the authors study the $\msrdc$ (Minimum Sum of Radius-Dependent Costs) problem: given $k$ and some metric space $(V,d)$ where $V=F\cup C$ for facilities $F$ and clients $C$, the goal is to find a clustering given by $k$ facility-radius pairs $(f_1,r_1),\ldots ,(f_k,r_k)$ such that all clients are covered, i.e., $C\subseteq B(f_1,r_1)\cup \cdots \cup B(f_k,r_k)$ where $B(x,r)$ is the radius-$r$ ball centered at $x$, and such that $\sum_{i=1}^k g(r_i)$ is minimized for some increasing function $g$. Their main contribution is that the $\msrdc$ problem in shortest-path metrics can be solved in polynomial time for bounded-treewidth graphs (see also~\cite{swat2026} for related results on this class of problems). If we want to solve the 
${\ensuremath\mathsf{Em\text{-}VC}}$ problem on a graph $G=(V,E)$ we can use $G'$, the expanded version of $G$ (see Definition~\ref{def-expanded-graph}) and then solve the $\msrdc$ on $G'$ with the metric induced by the shortest path distances in $G'$, with $F$ (resp. $C$) the vertices of $G$ (resp. $G'$), and with $k$ ranging from $1$ to $|V|$. Notice, however, that this reduction is valid only if the weights in graph $G$ satisfy the triangle inequality. This is of course the case if $G$ is a tree.  This shows that the problem is solvable in polynomial time on trees. However, we prove this result with a different approach in Theorem~\ref{theo-tree}, and the algorithm we obtain has a much lower time complexity.


\textbf{Our contribution.} \quad In Section~\ref{sec:def} we introduce the basic definitions, we show as a warm-up that the problem is solvable in polynomial time if we assume unit weights, and we develop the main lower bound used in the analysis of our approximation algorithm. We show that the problem is $\NP$-hard in Section~\ref{sec:hardness}. Section~\ref{sec:continuous} presents a polynomial-time algorithm for the continuous variant. Building on this formulation, Section~\ref{sec:approx} gives a $\tfrac{4}{3}$-approximation algorithm for the discrete problem. Finally, Section~\ref{sec:polynomial} identifies polynomial-time solvable cases, including complete graphs and trees (via a reduction to the $\mincover$ problem~\cite{Kolen}).

\section{Definitions and preliminary results}\label{sec:def}

We consider a simple, undirected, connected graph $G=(V, E)$ with edge weights $w:E\to\bbN$, where $\bbN=\{0,1,2,\ldots\}$. Throughout the paper, we use both the \emph{weighted length} and the \emph{number of edges} of a path. For a path $q=(e_1,\dots,e_{\ell})$, we write $|q|={\ell}$ for the number of edges and $w(q)=\sum_{i=1}^{\ell} w(e_i)$ for its weighted length. The distance $d(u,v)$ always denotes the weighted length of a shortest path between $u$ and $v$. \textcolor{black}{Moreover, when considering two graphs simultaneously, we write $d_G(u,v)$ to avoid ambiguity.} We use $\diam(G)$ for the diameter of~$G$.

A \emph{power assignment} on~$V$ is a function $\power:V\to\bbN$. A node~$v$ \emph{covers} an edge~$e$ if an agent that starts 
from~$v$ with power~$\power(v)$ has enough power to traverse~$e$, assuming that the agent moves only along edges of the graph, and that for each edge traversal, it spends power equal to the 
weight of the traversed edge. We denote by $\cost{v}(e)$ the power required at a node~$v$ to cover an edge $e=\edge{x}{y}$, i.e., \(\cost{v}(e)=\min\bigl(d(v,x),d(v,y)\bigr)+w(e)\).
With this definition, node~$v$ \emph{covers} edge~$e$ if and only if $\power(v)\geq \cost{v}(e)$. 

\begin{observation} \label{obs:reachendpoints}
    If a node~$v$ covers edge~$\{x,y\}$, then it has sufficient power to reach both~$x$ and~$y$. We have, then, $\power(v)\geq d(v,x)$ and $\power(v)\geq d(v,y)$. The converse is not necessarily true.
\end{observation}

In the \textsf{Emergency Vertex Cover} problem, the goal is to compute a power assignment~$\power$ so that every edge is covered by some node, and the total power on all nodes is minimized. The 
problem is defined formally as follows:

\begin{problem}[$\probname, \probnameshort$] $ $ \\
  \emph{Instance}: $\langle G, w\rangle$, where $G=(V, E)$ is a simple, undirected, connected, weighted graph with edge weights $w: E\to\bbN$. \\
  \emph{Feasible solution}: a power assignment $\power:V\to\bbN$ such that for every~$e\in E$, there exists~$v\in V$ with $\power(v)\geq \cost{v}(e)$. \\
  \emph{Goal}: minimize $\sum_{v\in V} \power(v)$.
\end{problem}

We assume without loss of generality that $w$ is everywhere strictly positive. Indeed, a zero-weight edge can be handled simply by contracting the edge and, if parallel edges are created, keeping only the heaviest one between any pair of vertices, since if the heaviest edge is covered, then all the others are also covered.

The \emph{support} of a solution is the set of nodes that receive nonzero power. An \emph{optimal solution with minimum support} is an optimal 
solution with support~$S$, such that the support $S'$ of every other optimal solution satisfies $|S'| \geq |S|$. 
A path~$q=(e_1,\dots,e_{\ell})$ is a \emph{covering path from~$v\in V$ to~$e_{\ell}\in E$} if $q$ starts from~$v$, its last edge is~$e_{\ell}$, and $w(q)\leq \power(v)$.



\subsection{Unit weights}


We show that, for the special case of graphs with unit weights, there always exists an optimal solution with support of size~$1$, which can be computed in polynomial time. We state Lemma~\ref{prop:mergesources} below for weighted graphs, and then we apply it to unit-weight graphs.

\begin{lemmarep}  \label{prop:mergesources}
Let $\power: V\to\mathbb{N}$ be a solution with support $S$, where $|S|\geq 2$, and let $u,v$ be two distinct vertices in $S$. Let $q=(e_1,\dots,e_\ell)$ be a path from $u$ to $v$ consisting of $\ell\ge 1$ edges, with no edge repeated.
 Let $A = \max \{ a\in [0,\ell] : \power(u) \geq \sum_{i=1}^a w(e_i)\}$, and $B =  \min \{ 
b\in[1,\ell+1] : \power(v) \geq \sum_{i=b}^\ell w(e_i)\}$.  %
If $B\leq A+1$ and there exist indices $\sigma\in[B,A+1]$ and $\lambda\in[1,\ell+1]$ such that $\sum_{i=1}^{\sigma-1} w(e_i) = \sum_{i=\lambda}^\ell w(e_i)$, then there 
exists another feasible solution with the same cost and strictly smaller support.
\end{lemmarep} 

\begin{proof}
 If $\lambda = \ell+1$, then $\sum_{i=1}^{\sigma-1} w(e_i) = \sum_{i=\ell+1}^\ell w(e_i)=0$, therefore $\sigma\leq 1$. This implies $B=1$ because 
$B\leq \sigma\leq 1$ and $B\geq 1$ by definition of~$B$. We construct a solution $p'$ with $\power'(v)=\power(u)+\power(v)$, $\power'(u)=0$, and $\power'(z)=\power(z)$ for 
$z\notin\{u,v\}$.
This solution has strictly smaller support and it is feasible because, starting from~$v$, one can follow $q$ in reverse 
and arrive at $u$ with remaining power at least $\power(u)+\power(v)-\sum_{i=1}^\ell w(e_i)=\power(u)+\power(v)-\sum_{i=B}^\ell w(e_i) \geq \power(u)$, where we used the fact that 
$\sum_{i=B}^\ell w(e_i) \leq \power(v)$ by definition of~$B$.

If $\lambda\in[1,\ell]$, note that $\sum_{i=1}^{\sigma-1} w(e_i) = \sum_{i=\lambda}^\ell w(e_i)$ implies $\sum_{i=\sigma}^{\ell} w(e_i) = 
\sum_{i=1}^{\lambda-1} w(e_i)$. Let $x$ be the node from which one starts the traversal of $e_\lambda$ while walking on the path~$q$. We construct a 
solution~$\power'$ with $\power'(x) = \power(x)+\power(u)+\power(v)$, $\power'(u)=\power'(v)=0$, and $\power'(z)=\power(z)$ for $z\notin\{u,x,v\}$. This solution has strictly smaller support and 
it is feasible because:
\begin{itemize}
 \item Starting from~$x$, one can follow the edges $(e_{\lambda-1},\dots,e_1)$
 and arrive at~$u$ with remaining power at least
\begin{align*}
\power(x)+\power(u)+\power(v)-\sum_{i=1}^{\lambda-1} w(e_i) 
&= \power(x)+\power(u)+\power(v)-\sum_{i=\sigma}^\ell w(e_i) \\
&\geq \power(x)+\power(u)+\power(v)-\sum_{i=B}^\ell w(e_i) \geq \power(u)
\end{align*}
 where we 
used the facts 
that $\sigma\geq B$ and   $\sum_{i=B}^\ell w(e_i) \leq \power(v)$ by definition of~$B$.
 \item Starting from~$x$, one can follow the edges $(e_{\lambda},\dots,e_\ell)$ and arrive at~$v$ with remaining power at least
\begin{align*}
\power(x)+\power(u)+\power(v)-\sum_{i=\lambda}^\ell w(e_i) &= \power(x)+\power(u)+\power(v)-\sum_{i=1}^{\sigma-1} w(e_i) \\
&\geq \power(x)+\power(u)+\power(v)-\sum_{i=1}^{A} w(e_i) \geq \power(v)
\end{align*} where 
we used the 
facts that $\sigma\leq A+1$ and $\sum_{i=1}^{A} w(e_i) \leq \power(u)$ by definition of~$A$.\qedhere
\end{itemize}
\end{proof}

Lemma~\ref{prop:mergesources} formalizes a merging principle: when two powered vertices $u$ and $v$ are sufficiently close along a simple path $q$, their power can be reallocated and combined without increasing the total cost. The proof (given in the appendix) proceeds by shifting the power of one vertex toward an appropriate point on $q$, thereby merging the two supports into a single one.

\begin{corollary} \label{cor:mergesources}
Let $\power: V\to\bbN$ be a solution with support~$S$, $|S|\geq 2$, and $u$, $v$ two distinct nodes in~$S$. If there exists a path~$q$ from~$u$ to~$v$ such 
that $\sum_{e\in q} w(e) \leq \max\bigl(\power(u),\power(v)\bigr)$, 
then there exists another solution with the same cost but with strictly smaller support.
\end{corollary}

\begin{proof}
Assume without loss of generality that $\power(u)\geq \sum_{e\in q} w(e)$. There exists a path~$\tilde{q}$ from~$u$ to~$v$ without repeated edges such that $\sum_{e\in\tilde{q}} w(e)\leq \sum_{e\in q} w(e)$. We apply Lemma~\ref{prop:mergesources} for $u$, $v$, and $\tilde{q}$. We have 
$A=|\tilde{q}|$ and therefore $B\leq A+1$ is satisfied. The remaining hypotheses of Lemma~\ref{prop:mergesources} hold by taking 
$\sigma=A+1$ and $\lambda=1$. Thus, the corollary follows.
\end{proof}


\begin{observation} \label{obs:unit:mergesources}
In a graph with unit weights, let $\power: V\to\bbN$ be a solution with support $S$, $|S|\geq 2$, $u,v$ be two distinct nodes in~$S$, and $q$ be a path 
with no repeated edges from~$u$ to~$v$. The hypotheses of Lemma~\ref{prop:mergesources} are satisfied by $u$,$v$, and $q$ if and only if 
$\power(u)+\power(v)\geq |q|$. 
\end{observation}

Lemma~\ref{ref:unit} follows by repeatedly applying Lemma~\ref{prop:mergesources}, using Observation~\ref{obs:unit:mergesources} to express the merging condition in the unit-weight setting.

\begin{lemma}\label{ref:unit}
In a graph with $n\geq 2$ nodes and unit weights, an optimal solution with support of size~$1$ exists.
\end{lemma}

\begin{proof}
  Let $\power: V\to\bbN$ be an optimal solution with minimum support, and let $S$ be the support of~$\power$. For a contradiction, suppose that $|S|\geq 2$.
Let~$u,v$ be two distinct nodes in~$S$ and let~$q$ be a shortest path from~$u$ to~$v$. If $\power(u)+\power(v)\geq |q|$, then by 
Observation~\ref{obs:unit:mergesources}, Lemma~\ref{prop:mergesources} applies, and we obtain an optimal solution with support size smaller 
than~$|S|$, a contradiction.

On the other hand, if $\power(u)+\power(v)<|q|$, then let $e=\edge{x}{y}$ be the $\bigl(\power(u)+1\bigr)$-st edge of~$q$ that one traverses while following~$q$ 
from~$u$ to~$v$. This edge must be covered by some node $z\notin\{u,v\}$ with $\power(z)\geq \min\bigl(d(z,x),d(z,y)\bigr)+1$. Note that, because of the 
unit weight assumption, $\min\bigl(d(z,x),d(z,y)\bigr) \geq d(z,x)-1$. Thus, $\power(z)\geq d(z,x)$.
Let $r$ be a shortest path from~$z$ to~$x$, let~$x'$ 
be the first node of~$r$ that is also in the sub-path of~$q$ from~$u$ to~$x$, and let~$s$ be the concatenation of the sub-path of~$q$ from~$u$ to~$x'$ 
and the reverse of the sub-path of~$r$ from~$z$ to~$x'$. Now, $s$ is a path from~$u$ to~$z$ and $\power(u)+\power(z)\geq |s|$, therefore 
Lemma~\ref{prop:mergesources} applies, and we obtain an optimal solution with support size smaller than~$|S|$, a contradiction.
\end{proof}

By Lemma~\ref{ref:unit}, it suffices to try every vertex as the unique powered vertex. A breadth-first search from each vertex gives all distances in $\calO(n(n+m))=\calO(nm)$ time, since the graph is connected, and the power needed at a vertex to cover all edges is then computed in $\calO(m)$ time. This yields the following corollary.

\begin{corollary}
In a graph with $n\geq 2$ nodes and unit weights, an optimal solution is to choose $v^\star \in \argmin_{v\in V}\ \max_{e\in E} \cost{v}(e)$ and set
$\power(v^\star)= \max_{e\in E} \cost{v^\star}(e)$ and $\power(v)=0$ for $v\neq v^\star$. This solution can be computed in $\calO(nm)$ time, where $m=|E|$.
\end{corollary}

\subsection{A lower bound on the cost of feasible solutions}\label{sec:lower_bound}

We derive Lemma~\ref{lem:diamlb}, which is the main structural lemma that allows us to obtain a lower bound on the cost of an optimal solution. We use this in our analysis of the approximation ratio of the algorithm in Section~\ref{sec:approx}.

\begin{definition}[see Figure~\ref{fig:expanded}] \label{def-expanded-graph}
Let $G=(V, E)$ be a weighted graph with weights $w: E\to\bbN$. We define the corresponding \emph{expanded graph} $G'=(V',E')$ with weights 
$w':E'\to\bbR^+$ as follows: $V'=V \cup \bigcup_{e\in E} \left\{h_e,h'_e\right\}$, $E'=\bigcup_{\edge{u}{v}\in E} \bigl\{\edge{u}{h_{\edge{u}{v}}}, \edge{v}{h_{\edge{u}{v}}}, 
\edge{h_{\edge{u}{v}}}{h'_{\edge{u}{v}}}\bigr\}$, and, for each $\edge{u}{v}\in E$, $w'\bigl(\edge{u}{h_{\edge{u}{v}}}\bigr) = w'\bigl(\edge{v}{h_{\edge{u}{v}}}\bigr) = 
w'\bigl(\edge{h_{\edge{u}{v}}}{h'_{\edge{u}{v}}}\bigr) = \frac{1}{2}w\bigl(\edge{u}{v}\bigr)$.
\end{definition}

\begin{figure}[ht]
\centering
\begin{tikzpicture}[scale=1]
  \node[draw,circle] (u) at (0,0) {$u$};
  \node[draw,circle] (v) at (2,0) {$v$};
  \draw[draw,circle] (u) -- (v) node[midway, above] {$w(e)$};

  \node[draw,circle] (u2) at (5,0) {$u$};
  \node[draw,circle] (v2) at (7,2) {$h_{e}$};
  \node[draw,circle] (c) at (7,0) {$h'_{e}$};
  \node[draw,circle] (x) at (9,0) {$v$};
  
  \draw (u2) -- (c) node[midway, above] {$\tfrac{1}{2}w(e)$};
  \draw (v2) -- (c) node[midway, right] {$\tfrac{1}{2}w(e)$};
  \draw (c) -- (x) node[midway, above] {$\tfrac{1}{2}w(e)$};
  
\end{tikzpicture}
\caption{Transformation of an edge $e=\edge{u}{v}$ in $G$ into its expanded version in $G'$.}
\label{fig:expanded}
\end{figure}
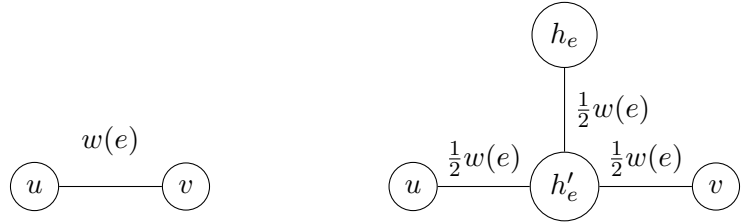

\begin{observation} \label{obs:expsol}
 Every feasible solution for~$\probnameshort$ on~$\langle G,w\rangle$ is also a feasible solution for~$\probnameshort$ on~$\langle G',w'\rangle$, where $G'$ is the corresponding expanded graph with weights~$w'$.
\end{observation}

By construction of $G'$ (each edge $e=\edge{u}{v}$ replaced by a  three-edge star with
weights $w(e)/2$), the shortest-path distances between the original vertices are preserved:
for all $a$, $b\in V$, $d_{G}(a,b)=d_{G'}(a,b)$. 
Hence, we can deduce the following observation (in some graphs this bound is tight, for instance when $G$ is a single edge).

\begin{observation} \label{obs:expdiam}
 If $\langle G',w'\rangle$ is the corresponding expanded graph of~$\langle G,w\rangle$, then $\diam(G') \geq \diam(G)$.
\end{observation}

\begin{lemma} \label{lem:diamlb}
 Let $\power: V\to\bbN$ be a feasible solution for~$\probnameshort$ on~$\langle G,w\rangle$, where $G=(V,E)$. Fix a node~$u\in V$, and let $E_u\subseteq E$ be the set of 
edges that are covered only by~$u$.
Let $\langle G',w'\rangle$ be the expanded graph of $\langle G\setminus E_u, w\rangle$,  and let $q$ be a shortest path between an arbitrary pair of vertices within a connected component 
of~$G'$. Then, 
$\sum_{v\in V} \power(v) \geq \power(u)+\frac{1}{2}w'(q)$.
\end{lemma}

\begin{proof}
Let~$H'$ be the connected component of~$G'$ that contains~$q$ and let~$H$ be the corresponding connected component of~$G\setminus E_u$. Let~$\power'$ be a power assignment defined on the nodes of~$H$ by taking the restriction of~$\power$ to the nodes of~$H$, and setting $\power'(u)=0$. Note that $u$ may be outside of~$H$, in which case~$\power'$ is simply the restriction of~$\power$ to the nodes of~$H$. 

\begin{claim} \label{clm:compH}
$\power'$ is a feasible solution for~$\langle H,w_H\rangle$, where $w_H$ is the restriction of~$w$ to the edges of~$H$. 
\end{claim}

\begin{claimproof}
Consider an arbitrary edge~$e$ of~$H$. By construction of~$G'$, we have $e\notin E_u$, hence in the original instance~$\langle G,w\rangle$, with power assignment~$\power$, $e$ must be covered by some node~$v\neq u$. Since $v$ covers every edge of every covering path from~$v$ to~$e$, it follows that all such covering paths are entirely in component~$H$. Moreover, since~$v\neq u$, we have $\power'(v)=\power(v)$. Therefore, $v$ also covers~$e$ in~$\langle H,w_H\rangle$, under power assignment~$\power'$.
\end{claimproof}

By Claim~\ref{clm:compH} and Observation~\ref{obs:expsol}, $\power'$ is also a feasible solution for~$\langle H',w'_{H'}\rangle$, where~$w'_{H'}$ is the restriction of~$w'$ to the edges in~$H'$.

Now, fix a  partition 
of~$q$ into sub-paths $q_1,\dots,q_\ell$, for some $\ell\geq 1$, with the properties that for each~$i$, the first and the last edge of~$q_i$ are covered by the same 
node~$v_i \neq u$ under~$\power'$, and also $v_i\neq v_j$ for $i\neq j$. We can construct a partition with these properties in the following manner: First, assign to each edge~$e$ of~$q$ a covering node~$v(e)$, choosing one arbitrarily if an edge is covered by multiple nodes. Note that~$v(e)\neq u$, since~$\power'(u)=0$ (or $u$ is not even in~$H$). Let~$e$ be the first edge of~$q$, and let~$e'$ be the last edge of~$q$ with~$v(e')=v(e)$. The first path~$q_1$ of the partition is then the sub-path of~$q$ from~$e$ up to and including~$e'$. In the remaining part of~$q$, all edges are covered by nodes different from~$v(e)$. Clearly, then, repeated application of the above procedure for the remaining part of~$q$ yields a partition with the required properties.
 
To conclude the proof, for a fixed sub-path~$q_i$ with endpoints~$s_i,t_i$, we know that $d(v_i,s_i)\leq \power'(v_i)$ and $d(v_i,t_i)\leq 
\power'(v_i)$ by Observation~\ref{obs:reachendpoints}. 
We conclude that $w'(q_i) = d(s_i,t_i) \leq d(s_i,v_i)+d(v_i,t_i) \leq 2\power'(v_i)$, where the distances refer to distances in~$\langle G',w'\rangle$. 
Therefore, summing over all sub-paths, $w'(q) 
= \sum_{i=1}^\ell w'(q_i) \leq 
2\sum_{i=1}^\ell \power'(v_i) = 2\sum_{i=1}^{\ell} \power(v_i)$. Since all of the $v_i$'s are distinct and different from~$u$, we have:
\begin{equation*}
\sum_{v\in V}\power(v) \geq \power(u) + \sum_{i=1}^{\ell} \power(v_i) \geq \power(u) + \frac{1}{2}w'(q) \enspace.\qedhere
\end{equation*}
\end{proof}

\section{Hardness}\label{sec:hardness}


\begin{theorem} 
    $\probnameshort$ is $\NP$-hard.
\end{theorem}

\begin{proof}
We show that $\SAT$ can be reduced to $\probnameshort$ in polynomial time. Let $\phi$ be a $\SAT$ instance with $n$ variables $\{x_1, x_2,\dots,x_n\}$ and $m$ clauses $\{c_1, c_2,\dots,c_m \}$, where  each variable appears in at least one clause. We construct an $\probnameshort$ instance~$I=\langle G,w\rangle$ as follows (see Figures~\ref{fig:clause} and \ref{fig:variable} for an illustration):
\begin{enumerate}
    \item For each clause $c_j$, we define two vertices~$c_j$, $c'_j$ and an edge $\edge{c_j}{c'_j}$ between these vertices with weight $w\left(\edge{c_j}{c'_j}\right)=1$.
    \item For each variable $x_i$, we define two vertices $v_i$, $v'_i$ to represent the positive occurrence, and an edge  $\edge{v_i}{v'_i}$ between these vertices with weight $w\left(\edge{v_i}{v'_i}\right)=2^{i-1}\cdot n$.
    \item For each variable $x_i$, we define two vertices $\bar{v}_i$, $\bar{v}'_i$ to represent the negative occurrence, and an edge $\edge{\bar{v}_i}{\bar{v}'_i}$ between these vertices with weight $w\left(\edge{\bar{v}_i}{\bar{v}'_i}\right)=2^{i-1}\cdot n$.
    \item For each variable~$x_i$, $G$ also contains the edges $\edge{v_i}{ \bar{v}_i}$ and $\edge{v'_i}{\bar{v}'_i}$, both with weight~$1$.
    \item For every variable~$x_i$ and every clause~$c_j$, $G$ contains an edge $\edge{c_j}{v_i}$ with weight $2^{i-1}\cdot n$ if and only if $x_i$ appears positively in $c_j$.
    \item For every variable~$x_i$ and every clause~$c_j$, $G$ contains an edge $\edge{c_j}{\bar{v}_i}$ with weight $2^{i-1}\cdot n$ if and only if $x_i$ appears negatively in $c_j$.
\end{enumerate}


\begin{figure}[tbp]
\begin{minipage}[t]{0.54\linewidth}
    \centering
\begin{tikzpicture}[scale=0.8,
  v/.style={circle,fill,inner sep=1.6pt},
  lab/.style={fill=white,inner sep=1pt,font=\small}]
\node[v,label=below:{$c_j$}]  (c)  at (0,0.5) {};
\node[v,label=above:{$c'_j$}] (cp) at (0,2.5) {};
\draw (c) -- (cp) node[lab,midway,right] {$1$};
\node[font=\small,align=center] at (0,-1.1) {clause gadget\\ for $c_j$};
\node[v,label=above left:{$v_i$}]        (v)   at (4,2.5)   {};
\node[v,label=above right:{$\bar v_i$}]  (vb)  at (6.5,2.5) {};
\node[v,label=below left:{$v'_i$}]       (vp)  at (4,0)     {};
\node[v,label=below right:{$\bar v'_i$}] (vbp) at (6.5,0)   {};
\draw (v) -- (vb)   node[lab,midway,above] {$1$};
\draw (vp) -- (vbp) node[lab,midway,below] {$1$};
\draw (v) -- (vp)   node[lab,midway,left]  {$2^{i-1}\cdot n$};
\draw (vb) -- (vbp) node[lab,midway,right] {$2^{i-1}\cdot n$};
\node[font=\small,align=center] at (5.25,-1.1) {variable gadget for $x_i$\\ (vertex set $V_i$)};
\end{tikzpicture}
    \caption{Clause and variable gadgets used in the reduction from~$\SAT$ to~$\probnameshort$.}
    \label{fig:clause}  
\end{minipage}\hfill
\begin{minipage}[t]{0.44\linewidth}
    \centering
\begin{tikzpicture}[scale=0.7,
  v/.style={circle,fill,inner sep=1.6pt},
  lab/.style={fill=white,inner sep=1pt,font=\small}]
\node[v,label=above:{$c'_j$}] (cp) at (1.6,6) {};
\node[v,label=below:{$c_j$}]  (c)  at (1.6,3) {};
\draw (cp) -- (c) node[lab,midway,left] {$1$};
\foreach \i/\y in {a/6.9,b/3.7,c/0.5}{
  \node[v,label=above:{$v_{\i}$}]      (v\i)  at (5.5,\y) {};
  \node[v,label=above:{$\bar v_{\i}$}] (vb\i) at (7.5,\y) {};
  \draw (v\i) -- (vb\i) node[lab,midway,below] {$1$};
}
\draw (c) -- (vba)  node[lab,pos=0.55] {$2^{a-1}\cdot n$};
\draw (c) -- (vb) node[lab,pos=0.42] {$2^{b-1}\cdot n$};
\draw (c) -- (vc)  node[lab,pos=0.55] {$2^{c-1}\cdot n$};
\end{tikzpicture}

  \caption{Connections between the clause gadget of $c_j=(\lnot x_a\lor x_b\lor x_c)$ and the variable gadgets of its literals.}
    \label{fig:variable}  
\end{minipage}
\end{figure}

The number of vertices and edges in $G$ is $\mathcal{O}(n+m)$. Furthermore, since the edge weights are bounded by $n2^n$, their binary encoding requires at most $\mathcal{O}(n+\log n)$ bits. Hence, the weighted graph $G$ can be constructed in polynomial time under the bit complexity model.

We claim that $\phi$ is a satisfiable formula if and only if the $\probnameshort$ instance~$I$ has an optimal solution with total power at most $2^{n} \cdot n$.

Indeed, fix a satisfying assignment of~$\phi$ and consider a power assignment~$\power$ on the nodes of~$G$, defined as follows: For each~$i$,
\begin{equation*}
    \power(v_i)=\begin{cases}
        2^{i-1}\cdot n + 1 & \text{if $x_i$ is true} \\
        0 & \text{otherwise}
    \end{cases} \text{\quad\ and\quad}
    \power(\bar{v}_i)=\begin{cases}
     2^{i-1}\cdot n + 1 & \text{if $x_i$ is false} \\
     0 & \text{otherwise}
    \end{cases} \enspace.
\end{equation*}
All other nodes $v$ have zero power: $\power(v)=0$. It is easy to check that $\power$ is a feasible solution for $I$ with total power $2^{n} \cdot n$: The node~$v_i$ or~$\bar{v}_i$ covers its entire variable gadget, including the edges connecting it to the clause gadgets, plus the clause gadgets of the clauses that are satisfied by variable~$x_i$. Since this is a satisfying assignment, every clause is satisfied by some variable; therefore, every edge of~$G$ is covered.

For the converse, let~$\power$ be an optimal solution of~$I$ with total power at most $2^n \cdot n$. For each index~$i$ ($1\leq i\leq n$), we define the vertex set $V_i=\{v_i, v'_i, \bar{v}_i, \bar{v}'_i\}$ and denote three of the edges of the corresponding variable gadget as $e_i=\edge{v'_i}{\bar{v}'_i}$, $\mathsf{pos}_i=\edge{v_i}{v'_i}$, and $\mathsf{neg}_i=\edge{\bar{v}_i}{\bar{v}'_i}$. We call $\mathsf{pos}_i$ and $\mathsf{neg}_i$ the \emph{critical edges} of the gadget.  By construction of the graph~$G$, $V_i$ is accessible from the outside only via an edge of weight~$2^{i-1}\cdot n$, connecting~$V_i$ to a clause gadget (see Figure~\ref{fig:variable}). Hence, we have the following property:
\begin{property} \label{prop:costsext}
For every~$i$ and for every node $v \notin V_i$, we have $\cost{v}(e_i) \ge 2^i \cdot n + 1$, $\cost{v}(\mathsf{pos}_i) \ge 2^i \cdot n$, and $\cost{v}(\mathsf{neg}_i) \ge 2^i \cdot n$.
\end{property}

 By Property~\ref{prop:costsext}, any vertex outside $V_i$ would require power at least $2^i n$
to cover the critical edges of the gadget. Claim~\ref{clm:reduction} states that exactly one vertex in each $V_i$ receives power $2^{i-1}n+1$, and all others receive zero.

\begin{claimrep} \label{clm:reduction}
    For every $i$, there exists exactly one vertex $v^\star_i\in V_i$ with $\power(v^\star_i) \geq 2^{i-1}\cdot n + 1$, 
and for all $x \in V_i \setminus \{v^\star_i\}$ we have $\power(x)=0$.
\end{claimrep}
\begin{proofsketch}
By Property~\ref{prop:costsext}, any vertex outside $V_i$ would require power at least $2^i n$ 
to cover the critical edges of the gadget. Since the total power is bounded by $2^n n$, we argue 
in decreasing order of $i$. For $V_n$, any external coverage would already consume the entire 
budget, so $V_n$ must be covered internally. Subtracting its exact contribution, the same 
argument applies inductively to $V_{n-1}, \dots, V_1$. In particular, covering the edges of weight $2^{i-1}n$ inside $V_i$ requires assigning power at 
least $2^{i-1}n+1$ to some vertex of the gadget. Moreover, because the global budget is exactly 
$\sum_{i=1}^n (2^{i-1}n+1)=2^n n$, exactly one vertex in each $V_i$ receives power $2^{i-1}n+1$, 
while all others receive zero, proving Claim~\ref{clm:reduction}.
\end{proofsketch}
\begin{proof} 
    We prove the claim by reverse induction on~$i$, starting with~$i=n$.
    
    We first observe that all three edges $e_n$, $\mathsf{pos}_n$, $\mathsf{neg}_n$ must be covered only from within~$V_n$. Indeed, by Property~\ref{prop:costsext}, if $v\notin V_n$ covers~$e_n$, then it must have $\power(v)\geq 2^n\cdot n+1$, which is more than the total power of~$\power$. Also by Property~\ref{prop:costsext}, if $v\notin V_n$ covers any of~$\{\mathsf{pos}_n,\mathsf{neg}_n\}$, then $\power(v)\geq 2^n\cdot n$, which means that all available power must be placed on~$v$ and hence~$e_n$ will be left uncovered.
    Furthermore, $\mathsf{pos}_n$ and $\mathsf{neg}_n$ must both be covered from the same node in~$V_n$, otherwise we are forced to have two nodes in~$V_n$, each with power $2^{n-1}\cdot n$, all other nodes of~$G$ with zero power, and, in particular, at least one edge~$\edge{c_j}{c'_j}$, where~$c_j$ is a clause that contains~$x_n$, that remains uncovered under~$\power$.
    
    It follows that exactly one node~$v^\star_n\in V_n$ must have power at least~$2^{n-1}\cdot n + 1$, since by the assumption on the cost of~$\power$, it is impossible to have more than one such node. Regardless of its exact position in~$V_n$, this node covers all four edges in the subgraph induced by~$V_n$.
    By optimality of~$\power$, all other nodes in~$V_n$ must have zero power. Indeed, any other node would have power at most~$2^{n-1}\cdot n - 1$, hence it would be unable to cover any new edge outside of the subgraph induced by~$V_n$.
    

    For the inductive step, fix $n_0\leq  n-1$ and assume the claim holds for all~$i\geq n_0+1$. Let $V^\star=\{v^\star_i:i \geq n_0+1\}$. By the inductive hypothesis, the total power on $V^\star$ is at least $\sum_{i=n_0+1}^{n} \bigl(2^{i-1}\cdot n + 1\bigr) = n\cdot (2^n-2^{n_0}) + n-n_0$. By the assumption on the cost of~$\power$, the available power that can be assigned to nodes outside of~$V^\star$ is $P\leq 2^n\cdot n - n\cdot (2^n-2^{n_0}) - n + n_0 = 2^{n_0}\cdot n - (n-n_0) \leq 2^{n_0}\cdot n-1$. 

   We now show that all three edges $e_{n_0}$, $\mathsf{pos}_{n_0}$, $\mathsf{neg}_{n_0}$ must be covered only from within~$V_{n_0}$. Indeed, suppose that $v\notin V_{n_0}$ covers any of $\{e_{n_0}, \mathsf{pos}_{n_0}, \mathsf{neg}_{n_0}\}$. By Property~\ref{prop:costsext} we must have $\power(v) \geq 2^{n_0}\cdot n$. Since $\power(v)>P$, $v$ must be in~$V^\star$. Suppose $v=v^\star_j$, for some $j\geq n_0+1$. By construction of~$G$, for $v^\star_j$ to cover~$e_{n_0}$ it must have power $\power(v^\star_j)\geq 2^{j-1}\cdot n + 2\cdot 2^{n_0-1}\cdot n + 1=2^{j-1}\cdot n + 2^{n_0}\cdot n + 1$. Therefore, the total power on~$V^\star$ is now at least $n\cdot (2^n-2^{n_0})+n-n_0 + 2^{n_0}\cdot n = n \cdot 2^n + n-n_0 \geq n \cdot 2^n + 1$, which is impossible. Moreover, for $v^\star_j$ to cover any of $\{\mathsf{pos}_{n_0}, \mathsf{neg}_{n_0}\}$, it must have power $\power(v^\star_j) \geq 2^{j-1}\cdot n + 2\cdot 2^{n_0-1}\cdot n=2^{j-1}\cdot n + 2^{n_0}\cdot n$. Therefore, the total power on~$V^\star$ is now at least $n\cdot (2^n-2^{n_0})+n-n_0 -1 + 2^{n_0}\cdot n = n \cdot 2^n + (n-1)-n_0$. If $n_0<n-1$, the total power on~$V^\star$ is greater than $n\cdot 2^n$, which is impossible. Therefore, the only way for $v^\star_j$ to cover one of $\{\mathsf{pos}_{n_0}, \mathsf{neg}_{n_0}\}$ is if $n_0=n-1$, in which case $v=v^\star_n$, $\power(v)=n\cdot 2^n$, and no other node is in the support of~$\power$, which leaves at least edge~$e_{n-1}$ uncovered.

   Additionally, $\mathsf{pos}_{n_0}$ and $\mathsf{neg}_{n_0}$ must both be covered from the same node in~$V_{n_0}$ since there is not enough available power to create two nodes with power~$2^{n_0-1}\cdot n$. It follows that exactly one node~$v^\star_{n_0}\in V_{n_0}$ must have power at least~$2^{n_0-1}\cdot n + 1$, and this node covers all four edges in the subgraph induced by~$V_{n_0}$. All other nodes in~$V_{n_0}$ must, therefore, have power at most~$P-\power(v^\star_{n_0}) \leq 2^{n_0}\cdot n - 1 - (2^{n_0-1}\cdot n+1) < 2^{n_0-1} \cdot n$, hence they are unable to cover any new edge outside of the subgraph induced by~$V_{n_0}$. By optimality of~$\power$, their power must be~$0$.
\end{proof}    

By Claim~\ref{clm:reduction}, the total power assigned to nodes in $\{v^\star_i: 1\leq i\leq n\}$ is 
\[
\sum_{i=1}^n \power(v^\star_i) \geq \sum_{i=1}^n (2^{i-1}\cdot n + 1) = 2^n\cdot n \enspace.
\]
As this is tight with respect to the assumed cost of~$\power$, all other nodes must have power~$0$ and all inequalities guaranteed by Claim~\ref{clm:reduction} must hold with equality. Therefore, the support of~$\power$ is exactly the set $\{v^\star_i: 1\leq i\leq n\}$ and, for every~$i$, $\power(v^\star_i) = 2^{i-1}\cdot n + 1$. 

Moreover, for every~$i$, $v^\star_i \in \{v_i,\bar{v}_i\}$. Indeed, by construction of~$G$, and given that $\power(v^\star_i) = 2^{i-1}\cdot n + 1$ for every~$i$, it follows that the edges connecting $v_i$ or $\bar{v}_i$ to the different clause nodes~$c_j$, for the clauses that contain variable~$x_i$, may only be covered by~$v^\star_i$. However, these edges are not covered if~$v^\star_i\in\{v'_i,\bar{v}'_i\}$.

Now, given~$\power$, we define the truth assignment in which $x_i$ is set to true if and only if $\power(v_i)>0$.
This assignment satisfies $\phi$. Indeed, consider any clause $c_j$. Since~$\power$ is a feasible solution, the edge~$\edge{c_j}{c'_j}$ must be covered. If it is covered by some node~$v_i$, then the edge~$\edge{v_i}{c_j}$ must be present, since otherwise node~$c_j$ is too far from~$v_i$. Therefore, by construction of~$G$, the variable~$x_i$ must appear positively in~$c_j$. By definition of the truth assignment, $c_j$ is satisfied. Similarly, if $\edge{c_j}{c'_j}$ is covered by some node~$\bar{v}_i$, then $x_i$ appears negatively in~$c_j$ and the clause is satisfied.
\end{proof}

\section{A continuous version of $\probnameshort$ ($\probnamecontshort$)}
\label{sec:continuous}

We introduce the continuous version of~$\probname$, in which power can be assigned at interior points of edges.

Given a simple, undirected, connected graph $G=(V,E)$ with edge weights $w:E\to\bbN$, consider a triple $(u,e,x)$, where $u\in V$, $e\in E$ is incident to~$u$, and $x\in\bbR^+$ satisfies $0\leq x\leq w(e)$. Such a triple represents a point \emph{inside} edge~$e$ at distance~$x$ from~$u$. Note 
that the point may coincide with one of the endpoints of~$e$ when $x=0$ or $x=w(e)$.  A single point may have multiple representations by different triples.

Let $\calP$ denote the set of points in~$G$.
We extend the definition of the shortest-path distance $d(a_1,a_2)$ to points $a_1,a_2\in\mathcal{P}$.  Specifically,  
if $a_1=(u,e_1,x)$ and $a_2=(r,e_2,y)$, where $e_1=\edge{u}{v}$ and $e_2=\edge{r}{z}$, then
\begin{equation*}
\begin{aligned}
d(a_1,a_2) = \min \Bigl\{ 
& x + d(u,r) + y,\:\:\: x + d(u,z) + (w(e_2)-y),  \\
& (w(e_1)-x) + d(v,r) + y,\:\:\: (w(e_1)-x) + d(v,z) + (w(e_2)-y)
\Bigr\}.
\end{aligned}
\end{equation*}

If $a=\bigl(u,\edge{u}{v},x\bigr)\in\calP$, then we define the power required at~$a$  to cover edge~$f=\edge{y}{z}$ as:
\begin{equation*}
\cost{a}(f) = 
\begin{cases}
\min\bigl(d(a,y),d(a,z)\bigr) + w(f), & \text{if } f\neq \edge{u}{v}, \\[6pt]
\max\bigl(d(a,u), d(a,v)\bigr),       & \text{if } f=\edge{u}{v}.
\end{cases}
\end{equation*}
In words, for $a$ to cover~$f$, it must have enough power to traverse~$f$ if $a$ is outside of~$f$, or to reach every point of~$f$ if $a$ is inside~$f$. Note that, if $x=0$, then $\cost{a}(f)=\cost{u}(f)$.

\begin{observation} \label{obs:cont:reachendpoints}
    An analogue of Observation~\ref{obs:reachendpoints} holds for~$\probnamecontshort$. If $a$ covers~$f=\{y,z\}$, then $\power(a)\geq d(a,y)$ and $\power(a)\geq d(a,z)$, regardless of whether $a$ is inside or outside of~$f$.
\end{observation}


We consider the continuous version of $\probname$, in which we 
seek to compute an assignment $\power:\calP\to\bbR^+$ of power with finite support $\calS=\{a\in\calP:\power(a)>0\}$, so that for every edge~$f$ there exists 
a point~$a$ with $\power(a)\geq \cost{a}(f)$.

\begin{problem}[$\probnamecont$, $\probnamecontshort$] $ $ \\
  \emph{Instance}: $\langle G, w\rangle$, where $G=(V, E)$ is a simple, undirected, connected, weighted graph with edge weights $w: E\to\bbN$. \\
  \emph{Feasible solution}: a power assignment $\power:\calP\to\bbR^+$, where $\calP$ is the set of points in~$G$, such that $\power$ has finite 
support~$\calS$ and for every~$e\in E$, there exists~$a\in \calP$ with $\power(a)\geq \cost{a}(e)$. \\
  \emph{Goal}: minimize $\sum_{a\in\calS} \power(a)$.
\end{problem}

In the remainder of this section, we analyze properties of optimal solutions with a focus on the size of their support to establish that the problem is solvable in polynomial time. Specifically, we show that there exists an optimal solution with support of size 1. Moreover, the single point in the support of that optimal solution belongs to a polynomial-time computable set of candidate points, which allows us to find an optimal solution in polynomial time.

\begin{lemmarep} \label{prop:cont:mergesources}
 Let $\power:\calP\to\bbR^+$ be an optimal solution for $\probnamecontshort$ on~$\langle G,w\rangle$, with support~$\calS$, $|\calS|\geq 2$, and let $a$, $b$ be two distinct 
points 
in~$\calS$ with $\power(a)+\power(b)\geq d(a,b)$. Then, there exists another optimal solution for $\probnamecontshort$ on~$\langle G,w\rangle$ with strictly smaller support.
\end{lemmarep}

\begin{proofsketch}
    We construct a new power assignment~$\power'$ by shifting the powers of $a$ and $b$ to a single point $c$ on the shortest path between them, setting $\power'(c) = \power(a) + \power(b)$ and $\power'(a)=\power'(b)=0$. This strictly reduces the support without increasing the total cost. To see that $\power'$ remains a feasible cover, we apply the triangle inequality: the combined power at $c$ provides sufficient budget to traverse the shortest path to either $a$ or $b$, and subsequently cover the remaining distance to any point previously reachable by them. 
\end{proofsketch}

\begin{proof}
Let $q$ be a shortest path between $a$ and $b$, and let~$c$ be the point on~$q$ at distance $\power(a)$ 
from~$b$, or $c=a$ if $\power(a) \geq d(a,b)$.

We define a power assignment~$\power':\calP\to\bbR^+$ as follows:
\begin{itemize} 
\item If $c=a$, then $\power'$ is identical to~$\power$, except that $\power'(a)=\power(a)+\power(b)$ and $\power'(b)=0$.
\item If $c\neq a$, then $\power'$ is identical to~$\power$, except that $\power'(c)=\max\{\power(c), \power(a)+\power(b)\}$ and $\power'(a)=\power'(b)=0$.
\end{itemize}
Note that $\power'$ has at most the same cost as~$\power$ and strictly smaller support. To show that $\power'$ is a feasible solution (and, therefore, optimal),    it suffices to show that every point reachable under~$\power$ from~$a$ (resp.\ $b$), can still be reached  under~$\power'$ from~$c$ by first going to~$a$ (resp.\ $b$). Indeed, this means that all covering paths from~$a$ (resp.\ $b$) are maintained under~$\power'$, with an extra prefix of a shortest path from~$c$ to~$a$ (resp.\ $b$). Note that the power assigned to~$c$ does not decrease, therefore the covering paths from~$c$ are always maintained under~$\power'$.

If $c=a$, let~$x$ be reachable from~$b$ under~$\power$, hence $d(b,x)\leq \power(b)$. Note that $c=a$ implies $d(a,b)\leq 
\power(a)$. We have, then: $d(a,b)+d(b,x) \leq \power(a) +  \power(b) =\power'(a)$. Therefore, $x$ can be reached under~$\power'$ from~$a$ by first going to~$b$.

If $c\neq a$, let $x$ be reachable from~$a$ under~$\power$, hence $d(a,x)\leq \power(a)$. Note that $c\neq a$ implies $d(c,b)=\power(a)<d(a,b)$. We have, then: $d(c,a) + d(a,x) = d(a,b)-d(c,b)+d(a,x) \leq \power(a)+\power(b) -\power(a)+\power(a) \leq  \power'(c)$.

Similarly, let $x$ be reachable from~$b$ under~$\power$, hence $d(b,x)\leq \power(b)$. We have, then: $d(c,b)+d(b,x) \leq \power(a) + \power(b) \leq \power'(c)$.
 \end{proof}

\begin{lemma} \label{prop:cont:onesource} 
There exists an optimal solution to $\probnamecontshort$ on~$\langle G,w\rangle$ whose support has size~$1$.
\end{lemma}

\begin{proof}
 Let $\power:\calP\to\bbR^+$ be an optimal solution with minimum support size, and let~$\calS$ be the support of~$\power$. For a contradiction, suppose that 
$|\calS|\geq 2$ and fix $a\in\calS$. 

\begin{claim}
    At least one node, different from~$a$, must be reachable from~$a$.
\end{claim}

\begin{claimproof}
    Since~$\power$ is optimal, $a$ must cover at least one edge~$e$, otherwise we could remove~$a$ from the support and obtain a feasible solution with smaller cost, contradicting the optimality of~$\power$. By Observation~\ref{obs:cont:reachendpoints}, both endpoints of~$e$ are reachable from~$a$. 
\end{claimproof}

\begin{claim}
    At least one edge is not covered by~$a$.
\end{claim}

\begin{claimproof}
    If $a$ covers all edges, then we can remove from the support every node other than~$a$ and obtain a feasible solution with smaller cost. This contradicts the optimality of~$\power$.
\end{claimproof}

By the above claims, let~$u\neq a$ be a node reachable from~$a$, and let $f=\{y,z\}$ be an edge that is not covered by~$a$. By Observation~\ref{obs:cont:reachendpoints}, at least one of the endpoints of~$f$ must be unreachable from~$a$. Without loss of generality, assume $y$ is unreachable from~$a$.

Since the graph is connected, on a path from~$u$ to~$y$ there must exist two adjacent nodes~$s,t$ such that $s$ is reachable from~$a$, hence $\power(a)\geq d(a,s)$, and $t$ is unreachable from~$a$. Moreover, by Observation~\ref{obs:cont:reachendpoints}, $a$ does not cover the edge~$\{s,t\}$. There must exist, then, some point $b\in\calS$ that covers~$\{s,t\}$, hence $\power(b)\geq d(b,s)$. By the triangle inequality, we have $d(a,b) \leq d(a,s)+d(b,s) \leq \power(a)+\power(b)$.
 Lemma~\ref{prop:cont:mergesources} applies, and we obtain an optimal solution with smaller support, which contradicts our choice of~$\power$.
\end{proof}

\begin{theoremrep} \label{prop:cont:solve} 
An optimal solution with support of size~$1$ for $\probnamecontshort$ can be computed in $\calO(m^4)$ time, where $m=|E|$.
\end{theoremrep}

\begin{proofsketch}
Given a $\probnamecontshort$ instance~$\langle G,w\rangle$ with $G=(V,E)$, we compute in polynomial time, inside every edge~$e$, a set of candidate points~$\calI(e)$. We then prove that one of the points 
in~$\bigcup_{e\in E}\calI(e)$ must be the support of an optimal solution.
The algorithm checks all candidates in 
$\bigcup_{e\in E}\calI(e)$ and outputs the one that needs the least power to cover all~edges. The set $\bigcup_{e\in E}\calI(e)$ has $\calO(m^3)$ points, and once all-pairs shortest paths are known, each point is evaluated in $\calO(m)$ time, hence the $\calO(m^4)$ bound.
\end{proofsketch}

\begin{proof}
Given an instance~$\langle G,w\rangle$ of $\probnamecontshort$ with $G=(V,E)$, we show how to compute in polynomial time, for every edge~$e$, a set of candidate points~$\calI(e)$ inside~$e$. We then prove that one of the points 
in~$\bigcup_{e\in E}\calI(e)$ is the support of an optimal solution with support of size~$1$.
The algorithm simply tries all points in
$\bigcup_{e\in E}\calI(e)$ and outputs  one that needs the minimum power to cover all edges.

 Fix an arbitrary edge $e=\edge{u}{v}$ and consider the function~$a:[0,w(e)]\to\calP$ with~$a(x)=(u,e,x)$. For an edge~$f\in E$, we are interested in~$\cost{a(x)}(f)$  as a function of~$x$. The following properties can be obtained by elementary arguments: 

 \begin{itemize}
    \item If $f=\edge{y}{z}\neq e$, we rewrite the 
covering cost of~$f$ from~$a(x)$ as:
\begin{align*}
 \cost{a(x)}(f) & = w(f) + \min\bigl(x+d(u,y),\, x+d(u,z),\, w(e)-x+d(v,y),\, w(e)-x+d(v,z) \bigr) \\
        & = w(f) + \min\Bigl(x+\min\bigl(d(u,y),d(u,z)\bigr),\, w(e)-x+\min\bigl(d(v,y),d(v,z)\bigr) \Bigr) \\
        & = \min\bigl(x+\cost{u}(f),\, w(e)-x+\cost{v}(f)\bigr)
\end{align*}
As a function of~$x$,  $\cost{a(x)}(f)$ is continuous, increasing at a slope of~$1$ for $x\in\bigl[0,  
x_f^\star\bigr)$, and decreasing at a slope of~$-1$ for $x\in\bigl(x_f^\star,w(e)\bigr]$, where $x_f^\star = \frac{1}{2} 
\bigl(\cost{v}(f)-\cost{u}(f)+w(e)\bigr) \in [0,w(e)]$. 

\item If $f=e$, we rewrite $\cost{a(x)}(e)$ as 
\(\cost{a(x)}(e) = \max\bigl(x, w(e)-x\bigr)\).
In this case, the function $\cost{a(x)}(e)$ is continuous, decreasing at a slope of~$-1$ for $x\in\bigl[0,x_e^\star\bigr)$, and increasing at a slope of~$1$ for 
$x\in\bigl(x_e^\star,w(e)\bigr]$, where $x_e^\star=\frac{1}{2}w(e)$. 
\end{itemize}

We conclude that for every~$f\in E$, the 
function~$\cost{a(\cdot)}(f):\bigl[0,w(e)\bigr]\to\bbR^+$ is fully specified by the values $\cost{a(0)}(f)=\cost{u}(f)\in\bbN$, $\cost{a(w(e))}(f)=\cost{v}(f)\in\bbN$, and 
$x_f^\star\in\bbQ$, which can all be computed exactly in polynomial time.

The collection of functions~$\calF_e=\bigl\{\cost{a(\cdot)}(f) : f\in 
E\bigr\}$ contains $|E|=\calO(|V|^2)$ functions. Moreover, in view of the particular form of the functions in~$\calF_e$, each pair of 
functions~$\mathfrak{f},\mathfrak{g}\in\calF_e$ has at most one crossing point\footnote{An $x_0\in\bigl(0,w(e)\bigr)$ is a \emph{crossing point} of 
$\mathfrak{f},\mathfrak{g}\in\calF_e$ if the function $\mathfrak{f}-\mathfrak{g}$ changes sign at~$x_0$, i.e., 
there exists $\epsilon>0$ such that $\mathfrak{f}-\mathfrak{g}$ is nonzero and has a constant sign in~$(x_0-\epsilon,x_0)$, is nonzero and has a constant sign in~$(x_0,x_0+\epsilon)$, and these signs are different. Note that, by continuity of $\mathfrak{f},\mathfrak{g}$, it must hold that $\mathfrak{f}(x_0)=\mathfrak{g}(x_0)$.} 
$x_{\mathfrak{f},\mathfrak{g}}^\star$ and this can be computed exactly in 
polynomial time from the values of~$\cost{u}(f)$, $\cost{v}(f)$, $x_f^\star$, $\cost{u}(g)$, $\cost{v}(g)$, and $x_g^\star$.

We define $\calI(e)=\{u,v\}\cup \bigl\{a(x_{\mathfrak{f},\mathfrak{g}}^\star):\mathfrak{f},\mathfrak{g}\in\calF_e\bigr\} 
\cup \bigl\{a(x_{f}^\star) : f\in E \bigr\}$, 
which contains $\calO(n^4)$ elements and can be computed in polynomial time.

\begin{claim}
    If $s$ is an interior point of~$e$ with $s\notin\calI(e)$, then there exists $s'\in\calI(e)$ with $\max_{f\in E} \cost{s'} (f) < \max_{f\in E}\cost{s}(f)$.
\end{claim}

\begin{claimproof}
Let $s=a(z)$, and consider two points $a(x_1),a(x_2)\in\calI(e)$ with $x_1<z<x_2$, such that no point of~$\calI(e)$ lies between $a(x_1)$ and $a(x_2)$. Such points must exist because $\calI(e)$ contains the endpoints of~$e$.

By construction of~$\calI(e)$, the set of edges that maximize $\cost{a(x)}(f)$ is the same for all $x\in(x_1,x_2)$. Let $\calE$ denote that set, i.e., $\calE=\argmax_{f\in E} \cost{a(x)}(f)$ for every 
$x\in(x_1,x_2)$. Recall that the cost functions change slope only at points $\bigl\{a(x_{f}^\star) : f\in E \bigr\} \subseteq \calI(e)$, and that none of these points can be between $a(x_1)$ and $a(x_2)$. Therefore, the functions in~$\bigl\{\cost{a(\cdot)}(f) : f\in 
\calE\bigr\}$ are either all increasing at a slope of~$1$ or all decreasing at a slope of~$-1$ in the interval~$(x_1,x_2)$. 
It follows that 
we can choose $z'\in\{x_1,x_2\}$ and have $\max_{f\in E} \cost{a(z')}(f) < \max_{f\in E} \cost{a(z)}(f)$.
\end{claimproof}

By the above claim and by Lemma~\ref{prop:cont:onesource}, there must exist an optimal solution whose support is a single point from the set~$\bigcup_{e\in E} \calI(e)$.
\end{proof}

\section{A \texorpdfstring{$\tfrac{4}{3}$}{4/3}-approximation algorithm for~$\probnameshort$}\label{ref:approx} \label{sec:approx}

We analyze the algorithm that exhaustively searches for the best solution with support of size~$1$. The algorithm selects a node $\tilde{v} \in 
\argmin_{v\in V} \ \max_{e \in E} (\cost{v}(e))$,  that is, 
the node requiring the minimum power (ties broken arbitrarily) to cover all edges on its own, and assigns to it power $\max_{e \in  E} \cost{\tilde{v}}(e)$. The power of every other node is set to~$0$.
For the remainder of this section, we consider a fixed instance $\langle G,w\rangle$ with $G=(V,E)$, and we adopt the following notation:
\begin{itemize}
 \item $\power^\star$ denotes an optimal power assignment for~$\probnameshort$ on~$\langle G,w\rangle$.
 \item $\OPT = \sum_{v\in V} \power^\star(v)$.
 \item $\power$ denotes the power assignment produced by the above algorithm on~$\langle G,w\rangle$.
 \item $\power_c$ is an optimal power assignment with support of size~$1$ for~$\probnamecontshort$ on~$\langle G,w\rangle$.
 \item $\hat{a}=(\hat{u}, \hat{e}, \hat{x})$ is the single 
point in the support of~$\power_c$, where, without loss of generality, we assume that $\hat{x}\leq \frac{1}{2} 
w(\hat{e})$.
\end{itemize}

To illustrate the $\tfrac{4}{3}$-approximation algorithm, consider the path drawn in Figure~\ref{ex:toto}. For the problem $\probnameshort$, 
an optimal solution assigns power only to vertices $b$ and $v$: $\power^\star(b)=1$ and $\power^\star(v)=2$, which covers all edges
($b$ covers $\edge{a}{b}$ and $\edge{b}{u}$, and $v$ covers the remaining ones), hence $ \OPT =3$. 
For the continuous variant $\probnamecontshort$, an optimal solution is attained at the single point $\alpha=(u,\edge{u}{v},1)$ with $\power_c(\alpha)=3$. 

The $\tfrac{4}{3}$-approximation algorithm computes $M(v)=\max_{e\in E}\cost{v}(e)$ for each vertex $v$. 
In this example, $M(a)=M(z)=6$, $M(v)=M(u)=4$, and $M(b)=5$. One possible output of the algorithm is to assign all power to the single vertex $u=\tilde{v}$, i.e., $\power(u)=4$.

\begin{figure}[h]
    \centering
\begin{tikzpicture}[scale=1]
\node[draw,circle,minimum size=5mm,inner sep=0pt] (a) at (0,0) {$a$};
\node[draw,circle,minimum size=5mm,inner sep=1pt,thick] (b) at (1.5,0) {$b$};
\node[draw,circle,minimum size=5mm,inner sep=0pt] (x) at (3,0) {$u$};
\node[draw,circle,minimum size=5mm,inner sep=1pt,thick] (y) at (4.5,0) {$v$};
\node[draw,circle,minimum size=5mm,inner sep=0pt] (z) at (6,0) {$z$};

\node (b1) at (1.5,0.7) {\small $\power(b)=1$};
\node (v1) at (4.5,0.7) {\small $\power(v)=2$};
 
\draw (a) -- (b) node[midway, above] {$1$};
\draw (b) -- (x) node[midway, above] {$1$};
\draw (x) -- (y) node[midway, above] {$2$};
\draw (y) -- (z) node[midway, above] {$2$};
\node (a1) at (3.75,-0.3) {\small \textcolor{white}{$\alpha$}};

\end{tikzpicture}\hfill
\begin{tikzpicture}[scale=1]
\node[draw,circle,minimum size=5mm,inner sep=0pt] (a) at (0,0) {$a$};
\node[draw,circle,minimum size=5mm,inner sep=0pt] (b) at (1.5,0) {$b$};
\node[draw,circle,minimum size=5mm,inner sep=0pt] (x) at (3,0) {$u$};
\node[draw,circle,minimum size=5mm,inner sep=0pt] (y) at (4.5,0) {$v$};
\node[draw,circle,minimum size=5mm,inner sep=0pt] (z) at (6,0) {$z$};
\node (v1) at (3.75,0.7) {\small $\power_c(\alpha)=3$};
\node[draw,fill,circle,minimum size=2mm,inner sep=0pt,color=black] (a1) at (3.75,0) { };
\node (a1) at (3.75,-0.3) {\small $\alpha$};
\draw (a) -- (b) node[midway, above] {$1$};
\draw (b) -- (x) node[midway, above] {$1$};
\draw (x) -- (y) node[midway, above] {$2$};
\draw (y) -- (z) node[midway, above] {$2$};

\end{tikzpicture}
\caption{Path of $4$ vertices. On the left: an optimal solution for $\probnameshort$, placing power on $b$ and $v$ ($\power^\star(b)=1$, $\power^\star(v)=2$). 
On the right: an optimal solution for the continuous variant $\probnamecontshort$, concentrated at an interior point $\alpha\in(u,v)$ with $\power_c(\alpha)=3$. }
\label{ex:toto}
\end{figure}



\begin{theorem}
The algorithm that selects a node
$\tilde{v} \in \argmin_{v\in V} \ \max_{e \in E} (\cost{v}(e))$
is a polynomial-time $\tfrac{4}{3}$-approximation algorithm for $\probnameshort$:
$\power(\tilde v)\le \tfrac{4}{3}\,\OPT.$
\end{theorem}

\begin{proof}

Let $v^\star$ be the node that covers edge~$\hat{e}$ under the optimal power assignment~$\power^\star$. Then, we  have $\power^\star(v^\star)\geq w(\hat{e}) \geq 2\hat{x}$.
Let $E^\star$ be the set of edges of~$G$ that are covered only by~$v^\star$ under~$\power^\star$. Let $\langle G',w'\rangle$ denote the expanded weighted graph corresponding to $\langle G\setminus E^\star,w\rangle$, and let $\langle G_{v^\star},w'\rangle$ be the connected component of~$\langle G',w'\rangle$ that contains the shortest path~$q_{v^\star}$ of maximum (finite) weight.


The proof is based on the following two claims, which provide lower bounds for $\OPT$ and upper bounds for the cost of~$\power$.

\begin{claim}
  $\OPT\geq \max\bigl\{ \power_c(\hat{a}),\, \power^\star(v^\star) + \frac{1}{2} w'(q_{v^\star}) \bigr\}$. 
\end{claim}

\begin{claimproof}
The first lower bound follows from the fact that every feasible solution of $\probnameshort$ is also a feasible solution of~$\probnamecontshort$ 
with the same cost. The second lower bound follows from Lemma~\ref{lem:diamlb}.
\end{claimproof}

\begin{claim}
 $\power(\tilde{v}) \leq \min\bigl\{ \power_c(\hat{a}) + \hat{x},\, \power^\star(v^\star) + w'(q_{v^\star}) \bigr\}$.
\end{claim}

\begin{claimproof}
 For the first upper bound, note that a power assignment $\power'$ with $\power'(\hat{u}) = \power_c(\hat{a})+\hat{x}$ and $\power'(v)=0$ for $v\neq\hat{u}$ is a 
solution with support of size~$1$ for $\probnameshort$ 
on~$\langle G,w\rangle$, as it covers all edges that were covered under~$\power_c$. Therefore,
\[
\power(\tilde{v})\leq \sum_{v\in V} \power'(v) = \power'(\hat{u}) = 
\power_c(\hat{a}) + \hat{x} \enspace.
\]

For the second upper bound, it suffices to prove that every edge~$e$ of~$G$ can be covered from~$v^\star$ with power at most~$\power^\star(v^\star) + 
w'(q_{v^\star})$. Let~$e$ be an arbitrary edge of~$G$. If $e\in E^\star$, then clearly $e$ can be covered from~$v^\star$ with power at 
most~$\power^\star(v^\star)$. Otherwise, the node~$h'_e$ must be in one of the components of~$\langle G',w'\rangle$. Recall that $h'_e$ is the node of the expanded graph associated with~$e$ (Definition~\ref{def-expanded-graph}). Since $G$ is connected, there must 
exist a node~$z\in V$ such that $z$ is at distance at most~$\power^\star(v^\star)$ from $v^\star$ and $z$ is in the same component of~$\langle 
G',w'\rangle$ as $h'_e$. By definition of~$q_{v^\star}$, the distance between $z$ and $h'_e$ in~$G'$ is at most~$w'(q_{v^\star})$, 
therefore
\ the edge~$e$ in~$G$ can be covered from~$z$ with 
power at most~$w'(q_{v^\star})$. It follows that $e$ can be covered from $v^\star$ with power at most~$\power^\star(v^\star) + 
w'(q_{v^\star})$.
\end{claimproof}

We now distinguish three cases, depending on whether the rounding loss~$\hat{x}$ is small compared to~$\power_c(\hat{a})$ (Case~1) or to~$\frac{1}{2}w'(q_{v^\star})$ (Case~2), or large compared to both (Case~3):

\noindent\emph{Case~1}: If $\hat{x}\leq \frac{1}{3} \power_c(\hat{a})$, then
\[
 \frac{\power(\tilde{v})}{\OPT} \leq \frac{\power_c(\hat{a}) + \hat{x}}{\power_c(\hat{a})} \leq \frac{\power_c(\hat{a}) + \frac{1}{3}\power_c(\hat{a})}{\power_c(\hat{a})} = \frac{4}{3} 
\enspace.
\]

\noindent\emph{Case~2}: If $\frac{1}{3} \power_c(\hat{a}) \leq \hat{x} \leq \frac{1}{2} w'(q_{v^\star})$, then, recalling that $\power^\star(v^\star)\geq 
2\hat{x}$, we have:
\[
 \frac{\power(\tilde{v})}{\OPT} \leq \frac{\power_c(\hat{a}) + \hat{x}}{\power^\star(v^\star) + \frac{1}{2} w'(q_{v^\star})} \leq \frac{\power_c(\hat{a}) + 
\hat{x}}{2\hat{x} + \frac{1}{2} w'(q_{v^\star})} \leq \frac{3\hat{x} + \hat{x}}{2\hat{x} + \hat{x}} = \frac{4}{3} \enspace.
\]

\noindent\emph{Case~3}: If $\hat{x} \geq \frac{1}{2} w'(q_{v^\star})$, then $\power^\star(v^\star) \geq 2\hat{x} \geq w'(q_{v^\star})$ and we have:
\[
 \frac{\power(\tilde{v})}{\OPT} \leq \frac{\power^\star(v^\star) + w'(q_{v^\star})}{\power^\star(v^\star) + \frac{1}{2} w'(q_{v^\star})} \leq \frac{w'(q_{v^\star}) 
+ 
w'(q_{v^\star})}{w'(q_{v^\star}) + \frac{1}{2} w'(q_{v^\star})} = \frac{4}{3} \enspace.
\qedhere\]
\end{proof}

Observe that the approximation ratio is tight: the instance depicted in Figure~\ref{ex:toto} reaches
$\power(\tilde{v})=\tfrac{4}{3}\,\OPT.$

\section{Polynomial-time solvable cases}\label{sec:polynomial}

In this section, we provide polynomial-time algorithms for two specific families of weighted graphs: complete graphs and trees.

\subsection{Complete graphs}
We first show that in a complete graph, some optimal solution is supported on a single vertex.

\begin{lemmarep}~\label{lem:completegraph}
In the complete graph $K_n$ with $n\geq 2$ and arbitrary edge weights, there exists an optimal solution whose support has size $1$.
\end{lemmarep}

\begin{proof}
Let $\power: V\to\mathbb{N}$ be an optimal solution with minimum support, and let $S$ denote the support of $\power$. For the sake of contradiction, assume that $|S|\geq 2$.
Let $u$ and $v$ be two distinct vertices in $S$.

 If $\max\bigl(\power(u),\power(v)\bigr) \geq w(\edge{u}{v})$, then Corollary~\ref{cor:mergesources} applies to $u$, $v$, and the path consisting of the single edge $\edge{u}{v}$. Hence, there exists an optimal solution whose support is strictly smaller than $|S|$.

If 
$\max\bigl(\power(u),\power(v)\bigr) < w\bigl(\edge{u}{v}\bigr)$, then the edge $\edge{u}{v}$ must be covered by some node~$z\notin\{u,v\}$ with $\power(z)\geq 
\min\bigl(d(z,u),d(z,v)\bigr)+w\bigl(\edge{u}{v}\bigr)$. 
Assuming, without loss of generality, that  $d(z,u)\leq d(z,v) $, it follows that  $\power(z)\geq 
d(z,u)+w\bigl(\edge{u}{v}\bigr) \geq d(z,u)$. Therefore, Corollary~\ref{cor:mergesources} applies to~$u$, $z$, and any shortest path between~$u$ and~$z$, and hence, there exists an optimal solution whose support is strictly smaller than $|S|$.

In both cases, we obtain an optimal solution with a support size smaller than $|S|$, a contradiction.
\end{proof}

As a direct consequence of Lemma~\ref{lem:completegraph}, the optimal solution in complete graphs can be computed in $\calO(n^3)$ time by trying every node and selecting one that needs the minimum power to cover all edges: all-pairs shortest paths take $\calO(n^3)$ time, and each of the $n$ candidates is then evaluated in $\calO(n^2)$ time.

 \begin{corollary}
In the complete graph $K_n$ with $n\geq 2$ and arbitrary edge weights, an optimal solution is obtained by selecting 
$v^\star \in \argmin_v \max_e \cost{v}(e)$, assigning $\power(v^\star) = \max_{e \in E} \cost{v^\star}(e)$, and setting $\power(v)=0$ for all $v\neq v^\star$. This solution can be computed in $\calO(n^3)$ time.

\end{corollary}

\subsection{Trees}

A naive algorithm would be to choose $v^\star \in \argmin_v \max_e \cost{v}(e)$, that is,  the node that requires the least power
 to cover all edges on its own (breaking ties arbitrarily), and assign it power equal to $\max_{e \in E} \cost{v^\star}(e)$. All other vertices would then receive power~$0$. 
However, this strategy is not always optimal. For example, consider a path with five edges whose weights, from left to right, are $a, b, c, b, a$, where $c \geq a+b$. Every single-support solution must deploy a power of at least~$c+a+b$, whereas a solution with total cost $c+\max\{a,b\}$ is available if we allow two nodes in the support.


In~\cite{Kolen}, the author considers the $\mincover$ ($\mincovershort$) problem, which is defined as follows: We are given a tree $T$ in which each edge has a positive weight, as well as a set of clients and facilities.   There is a client at every vertex of $T$, and we also assume that the facility locations are vertices of $T$.
Let $s_j$, $j=1,2,\ldots ,F$, denote the possible facility locations. A facility located at $s_j$ can serve only the clients located within distance $r_j\geq 0$ from $s_j$.
The cost of establishing a facility at $s_j$ is given by a constant $c_j$,   $c_j>0$. Note that several facilities may be located at the same vertex with different costs, and consequently different serving ranges, i.e., different sets of clients.
We say that facility $j$ is \emph{open} if there is a facility established at $s_j$.
The objective is to select a set of facilities to open such that all clients are served at minimum total cost. The problem can be solved in time $\calO(nF)$ where $n$ is the number of clients (i.e., number of vertices of $T$) and $F$ is the number of facilities~\cite{Kolen}. 







\begin{theoremrep} \label{theo-tree}
The $\probnameshort$ problem is solvable in $\calO(n^3)$ time on trees with $n$ vertices.
\end{theoremrep}

\begin{proofsketch}
    The idea is to show a polynomial-time reduction  from  the $\probnameshort$ problem on $T$ to the $\mincover$ ($\mincovershort$) problem on a tree $T'=(V',E')$ with additional vertices for each edge that we need to cover. Finally we use the algorithm in~\cite{Kolen} for the case of trees.
\end{proofsketch}

\begin{proof}
Let $T=(V,E)$ be a tree. We show a (polynomial-time) reduction  from  the $\probnameshort$ problem on $T$  to  the $\mincover$ ($\mincovershort$) problem on a tree $T'=(V',E')$.


Note that in the $\probnameshort$ problem the goal is to cover all edges, whereas in the minimum cost covering problem the goal is to cover all vertices. The main idea of the reduction is to create additional vertices in $T'$ so that covering all vertices in the minimum cost covering problem on $T'$ implies that all edges in $T$ are covered.

Specifically, we construct a tree $T'=(V',E')$ as follows:
\begin{itemize}
    \item for each vertex $v\in V$, we create a corresponding vertex $v$ in $V'$
    \item for each edge $e=\edge{x}{y}\in E$ of weight $w(e)$, we create two additional vertices $v_e$ and $v'_e$ in $V'$, and connect them with edges $\edge{x}{v_e}$, $\edge{v_e}{y}$ and $\edge{v_e}{v'_e}$ each one of weight $w(e)/2$ in $E'$. 
\end{itemize}


Once the tree $T'$ has been constructed, we define the set of facilities together with the corresponding vertices that each facility can cover.

First, observe that for each vertex $v \in V$, its power $\power(v)$ in any optimal solution must belong to the set of values ${\cal P}_v = \{\cost{v}(e) : e\in E\} = \{p_{v,1},\ldots ,p_{v,n_v}\}$, where $n_v$ denotes the number of distinct power values that $v$ can take. Note that $n_v \leq m$, since there are at most $m$ distinct distances between $v$ and the other vertices of the graph. 
We assume those values are sorted in increasing order, i.e., $p_{v,1}<\ldots < p_{v,n_v}$. A BFS algorithm can be applied to generate such sets. %
For each vertex $v'\in V'$, corresponding to a vertex $v'\in V$, we create $n_{v'}$ facilities located at $v'$.  Thus, the total number of facilities is $\calO(mn) = \calO(n^2)$. The cost of establishing the $i$-th facility at $v'$ is $p_{v',i}$,
for $1\leq i\leq n_{v'}$, and it can serve only those clients located within distance $p_{v',i}$ from $v'$. 
The $i$-th facility located at $v'$, together with the clients it can serve, forms a connected subtree denoted by $T'_{v',i}$.

\begin{figure}[htbp]
    \centering
    \includegraphics[width=\linewidth]{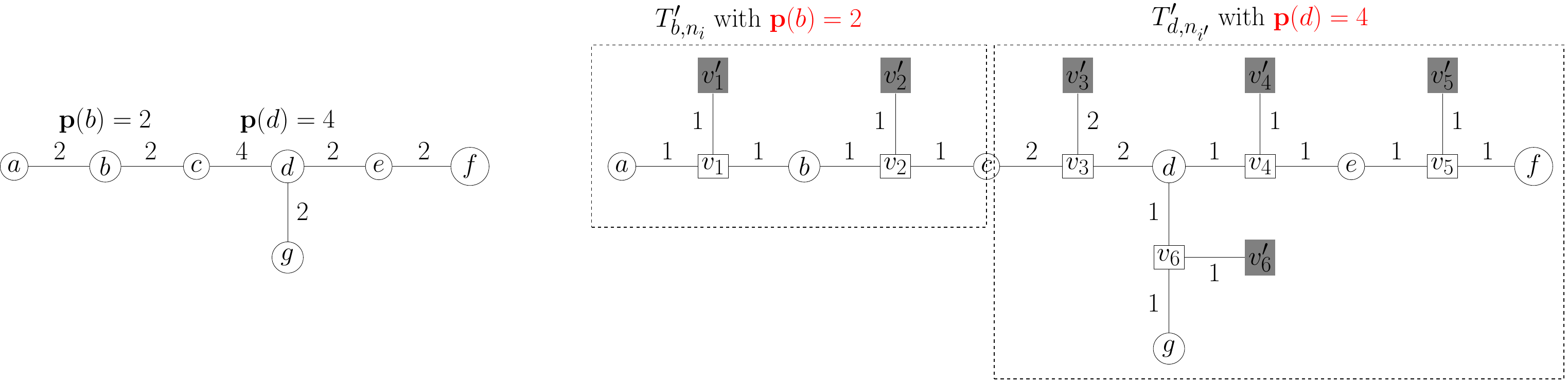}
    \caption{An illustration of the proof of Theorem~\ref{theo-tree} with ${\cal P}_b=\{2,6,8,10\}$ and ${\cal P}_d=\{2,4,6,8\}$. On the left (resp. right) we show the tree $T=(V,E)$ (resp. $T'=(V',E')$). Newly added vertices are drawn as rectangles, and the vertices $v'_i$ corresponding to edges $e \in E$ are shaded in gray. In both cases, the optimal solution has cost~$6$.
    } 
    \label{fig:enter-label}
\end{figure}




According to the algorithm in~\cite{Kolen}, the $\mincover$ ($\mincovershort$) problem on $T'$ can be solved in $\calO(nF)$ time, where $F$ is the number of facilities (with $F = \calO(n^2)$ in our case). Hence, the overall time complexity is $\calO(n^3)$.
We now show that the optimal solutions of the two problems are related.

$(\Rightarrow)$
Let us consider an optimal solution $s$ for the $\probnameshort$ problem on $T=(V,E)$. For each vertex $v\in V$ in the support of $s$, we denote by $\power^\star(v)$ its power. To construct a solution for $\mincovershort$, for each vertex $v'$, a facility is opened at $v'$ with power $\power^\star(v')$. Observe that $\power^\star(v')$ must necessarily coincide with one of the values $p_{v',i}$ for $1 \leq i \leq n_{v'}$; otherwise, the assigned power could be reduced, which would contradict the optimality of $s$.


Suppose an edge $e=\edge{x}{y}\in E$ is covered by a vertex $v$ with  power $\power(v)$. By definition, we have $\power(v)\geq \min\{d(v,x),d(v,y)\}+w(e)$. Without loss of generality,  assume that $\power(v)\geq d(v,x)+w(e)$. 
Since the edge $e$ can be covered from $x$ in $\probnameshort$, it follows that the vertices $v_{e}$, $v'_{e}$, $y$ can be covered from  $x$ in $\mincovershort$.
This holds because the distance from $x$ to each of $v_{e}$,  $v'_{e}$, and $y$ is at most $w(e)$.
Applying this argument to every edge shows that if an edge $e=\edge{x}{y}$ is covered in $\probnameshort$, then the corresponding vertices $v_{e}, v'_{e}, x,$ and $y$ are also covered in $\mincovershort$.
Observe that both solutions have the same cost. 

$(\Leftarrow)$ 
Let us consider a solution $s'$ to the $\mincovershort$ problem. We show how to transform it into a feasible solution to the $\probnameshort$ problem on $T=(V,E)$ with the same cost.
We focus on a facility opened at vertex $z$ in $s'$, with cost $p_{z,j}$ in $T'$. By construction of the reduction, node $z$ also has a corresponding counterpart in $T$.

Let $T'_{z,j}$ denote the subtree of $T'$ rooted at $z$, consisting of the facility located at $z$ together with the clients it can serve within cost $p_{z,j}$.

We focus on the edges in $T'_{z,j}$ and distinguish two cases:
\begin{itemize}
    \item  If $v'_i\in T'_{z,j}$ is a leaf, then  all  vertices $x$, $y$, $v_i$, and $v'_i$ are also covered. By construction, we know that $d(z,v'_i)\leq p_{z,j}$.
    Without loss of generality, assume that $d(z,x)<p_{z,j}$.  Then, we have $d(z,x)+d(x,v_i)+d(v_i,v'_i) = d(z,x)+d(x,v_i)+d(v_i,y) \leq p_{z,j}$.  Hence, the edge $\edge{x}{y}$ is covered by $z$ in $T$ using power $p_{z,j}$.
\item If $v_i \in T'_{z,j}$ is a leaf, then there exists another tree $T'_{z',j'}$ that covers vertex $v'_i$, and consequently also the vertices $x$, $y$, $v_i$, and $v'_i$, where $\edge{x}{y}\in E$.

\end{itemize}

\textcolor{black}{We denote by $k_{\power(v)}$ the index such that $p_{v,k_{\power(v)}}=\power(v)$.} 
From the above analysis, each vertex $v'_i$ is covered by some tree $T'_{z,j}$. 
Thus, by assigning power $\power(z)=p_{z,k_{\power(z)}}$ to every vertex $z\in V$ with $(z,p_{z,k_{\power(z)}})\in s'$, 
the corresponding edge $\edge{x}{y}\in E$ associated with $v'_i$ is also covered. 
Hence, we obtain a feasible solution for the $\probnameshort$ problem on $T=(V,E)$ with the same cost  $\sum_{z,\,(p_{z,k_{\power(z)}})\in s'} p_{z,k_{\power(z)}}$.
\end{proof}

\section{Conclusion}
We introduced the $\probname$ ($\probnameshort$) problem, established its $\NP$-hardness, derived general lower bounds via a distance-preserving expanded graph (including a $\diam(G)/2$ bound), gave a polynomial-time algorithm for the continuous variant with a support-one optimum, designed a $\tfrac{4}{3}$-approximation algorithm for the discrete variant, and identified polynomial-time solvable cases (complete graphs and trees via a reduction to the $\mincover$ problem). We highlight that our $\tfrac{4}{3}$-approximation algorithm concentrates all power on a single vertex.  This single-support structure is specific to unit weights, to the continuous relaxation, and to complete graphs; in general the optimum is distributed, as in the tree example of Section~\ref{sec:polynomial}. Promising directions include robust  extensions—for example, fixing the support size, or imposing robustness constraints such as requiring every edge to be covered via two internally disjoint paths or by two distinct covering vertices.

\bibliography{biblio}





\end{document}